\documentclass[10pt]{article}

\usepackage{amsfonts}
\usepackage{amsthm}
\usepackage{amssymb}
\usepackage{amsmath}
\usepackage{mathtools}
\usepackage{lscape}
\usepackage{graphicx}
\usepackage{epstopdf}
\usepackage{subfig}
\usepackage{color}
\usepackage{url}
\usepackage{vmargin}
\usepackage[parfill]{parskip}
\usepackage[round]{natbib}
\usepackage[noblocks]{authblk}
\definecolor{darkblue}{rgb}{0,0,0.6}
\usepackage[colorlinks=true,citecolor=darkblue]{hyperref}
 
\newtheorem{thm}{Theorem}%[section]
\newtheorem{definition}[thm]{Definition}

\numberwithin{equation}{section}

\begin{document}

\title{\bf The rate of convergence to early asymptotic behaviour in age-structured
epidemic models}
\author[1]{Christopher A. Rhodes}
\author[1,*]{Thomas House}
\affil[1]{Warwick Mathematics Institute, University of Warwick, Gibbet Hill
Road, Coventry, CV4 7AL, UK.}
\affil[*]{Corresponding Author: T.A.House@warwick.ac.uk}
\date{}

\maketitle

\begin{abstract}
Age structure is incorporated in many types of epidemic model. Often it is
convenient to assume that such models converge to early asymptotic behaviour
quickly, before the susceptible population has been appreciably depleted.  We
make use of dynamical systems theory to show that for some reasonable parameter
values, this convergence can be slow. Such a possibility should therefore be
considered when parameterising age-structured epidemic models.
\end{abstract}

\section{Introduction}

Age is one of the key variables to consider in any epidemiological analysis,
particularly related to infectious disease \citep{AndersonMay}.  In the context
of human diseases, age is an extrinsic risk factor, strongly influencing an
individual's mixing patterns within the population \citep{Keeling08}; but age is
also an intrinsic risk factor for serious outcomes and this can lead to
important and often counter-intuitive public-health
conclusions \citep{Anderson:1983}.

While classical models of childhood infections used \textit{a priori} matrices
to parameterise the epidemiological mixing between individuals of different
ages \citep{AndersonMay,Keeling08}, more recent studies have attempted to
measure this quantity empirically either using surveys
\citep{Edmunds06,Mossong08} or through the use of synthetic populations derived
from large datasets \citep{DelValle:2007}. The insights gained by doing this
were particularly useful during the 2009 influenza pandemic \citep{Baguelin10}.
A mathematical overview of the impact of recent empirical results for modelling
is given by \citet{Glasser2012}. 

In theoretical analysis of age-structured models, it is often convenient to
make the assumption that early behaviour of the system has converged to an
appropriately defined dominant eigenvector, representing the relative
prevalence of disease in different age groups.  Often this is done through the
consideration of discrete generations of infectious individuals, and the
relevant eigenvector is associated with a `next-generation matrix' whose
dominant eigenvalue is the basic reoproductive ratio $R_0$, which is the
expected number of secondary cases per primary case in a na\"{i}ve population
\citep{Diekmann:1990}.  While discrete generations are not always
straightforwardly related to the system's real-time behaviour, this framework
opens up the possibility of highly general analysis of a large number of
epidemiological scenarios \citep{Diekmann00}.  

Here, we analyse the rate of convergence of an age-structured epidemic to its
dominant eigenvalue in a real-time framework, using both analytic and numerical
methods.  This is done using a dynamical system of ordinary differential
equations (ODEs), together with results from linear algebra. We find that for
some plausible parameter values, it is possible for the timescale of
convergence to be comparable to the timescale over which susceptible depletion
becomes appreciable, and so the assumption of fast convergence to the dominant
eigenvector cannot be made.

\section{Age-structured model}

\subsection{Model definition}

Our modelling approach is based on the SIR compartmental structure in which
individuals are all either Susceptible, Infectious or Recovered, without
demographic processes like birth and death.  Individuals are also placed into
discrete age categories indexed by $a,b \in \mathbb{N}$, rather than using
continuous ages that would require a less tractable partial differential
equation or integro-differential equation model. Our ODE system is then
\begin{equation}
\begin{aligned} 
\frac{{\rm d}S_{a}}{{\rm d}t}\ & = -S_{a} \sum_{b} \beta_{a,b} I_{b} \text{ ,}\\
\frac{{\rm d}I_{a}}{{\rm d}t}\ & = S_{a} \sum_{b} \beta_{a,b} I_{b} - \gamma I_{a}
\text{ .} \label{SIR}
\end{aligned}
\end{equation}
Here the transmission matrix is $\mathbf{M} = (\beta_{a,b})$, where $
\beta_{a,b} $ is the rate of transmission from individuals in age class $b$
to individuals in age class $a$.

We will assume that there are $n$ age classes and adopt, for convenience, a
normalisation
\begin{equation}
S_a(t) + I_a(t) + R_a(t) = 1 \text{ ,} \quad \forall a,t \text{ .} \label{norm}
\end{equation}
This simplifies the algebra involved in analytic work (i.e.\ a more general
case can be considered in our framework at the cost of more complex and less
enlightening expressions) and can be justified for particular choices of age
classes.

Primarily our approach will look at analytically approximating the solutions of
our system near fixed points to determine the behaviour of infecteds and
susceptibles in the population. We will adopt a vector notation to represent
the dynamical state of the system.
\begin{equation}
\mathbf{S} = \left( \begin{array}{c} S_1 \\ \vdots \\ S_n \end{array} \right)
\text{ ,} \qquad
\mathbf{I} = \left( \begin{array}{c} I_1 \\ \vdots \\ I_n \end{array} \right)
\text{ ,} \qquad
\mathbf{x} = \left( \begin{array}{c} \mathbf{S} \\ \mathbf{I} \end{array} \right)
\text{ .}
\end{equation}
In this notation, we consider dynamical systems of the general form
\begin{equation}
\frac{{\rm d}\mathbf{x}}{{\rm d}t} = \mathbf{F}(\mathbf{x}) \text{ .}
\end{equation}
If $\mathbf{x}_*$ is a vector such that $\mathbf{F}(\mathbf{x}_*) =
\mathbf{0}$, the system dynamics can be linearised around this vector to give
\begin{equation}
\frac{{\rm d} \mathbf{y}}{{\rm d} t} = \mathbf{J} \mathbf{y} + O(\varepsilon)
\text{ ,} \qquad \text{where } \mathbf{x} = \mathbf{x}_* + \varepsilon \mathbf{y} \text{ .}
\end{equation}
In this regime, the dynamics are dominated by the eigenvalue of the Jacobian
matrix $\mathbf{J}$ with the largest real component. We call this the
\textit{dominant eigenvalue}, written $\lambda_1$, which has an associated
\textit{dominant eigenvector}, $\mathbf{v}_1$. Typically, there will be a rate
$r$ such that the ratio of the magnitude of $\mathbf{y}$ in the direction of
$\mathbf{v}_1$ to the magnitude of $\mathbf{y}$ in any orthogonal direction
will asymptotically grow like ${\rm e}^{rt}$.  We will have a quantity like $r$
in mind when discussing the rate of convergence of the system to its dominant
eigenvalue. If this rate is low compared to the rate at which $O(\varepsilon)$
effects become important, then we will say that the system fails to converge to
its dominant eigenvector. We note that this is different from the rate of
convergence of a stochastic model to its deterministic limit
\citep{Tuljapurkar:1982}, and is fundamentally a real-time concept, meaning
that direct analogues need not exist in a discrete generation-based framework
as analysed by \citet{Diekmann00}.

\subsection{Early dynamical behaviour}

In order to consider behaviour early in the epidemic, the Jacobian of
\eqref{SIR} will have to be considered, and is given in terms of quantities
defined in the previous section by
\begin{equation*}
\mathbf{J} = 
\left(
\begin{array}{ c c c c c c c c c }
- \sum_{b} \beta_{1,b} I_{b} & 0 & \cdots & 0 & -S_{1} \beta_{1,1} & \cdots & \cdots & -S_{1} \beta_{1,n} \\
0 & \ddots & \ddots & \vdots & \vdots & \ddots & & \vdots \\
\vdots & \ddots & \ddots & 0 & \vdots & & \ddots & \vdots \\
0 & \cdots & 0 & - \sum_{b} \beta_{n,b} I_{b} & -S_{n} \beta_{n,1} & \cdots & \cdots & -S_{n} \beta_{n,n} \\
\sum_{b} \beta_{1,b} I_{b} & 0 & \cdots & 0 & S_{1} \beta_{1,1} - \gamma & S_{1} \beta_{1,2} & \cdots & S_{1} \beta_{1,n} \\
0 & \ddots & \ddots & \vdots & S_{2} \beta_{2,1} & \ddots & \ddots & \vdots \\
\vdots & \ddots & \ddots & 0 & \vdots & \ddots & \ddots & S_{n-1} \beta_{n-1,n} \\
0 & \cdots & 0 & \sum_{b} \beta_{n,b} I_{b} & S_{n} \beta_{n,1} & \cdots & S_{n} \beta_{n,n-1} & S_{n} \beta_{n,n} - \gamma
\end{array}
\right).
\end{equation*}
We are interested in the behaviour around the disease-free equilibrium of
\eqref{SIR}, $ \mathbf{x}_{*} =(\mathbf{1}, \mathbf{0})^{\top} $, where
$\mathbf{1}$ and $\mathbf{0}$ are a length-$n$ vectors whose entries are all 1
and all 0 respectively. Putting this fixed point into the Jacobian gives
\begin{equation}
\mathbf{J}(\mathbf{x}_{*}) = 
\left(
\begin{array}{ c c }
\mathbf{0}& - \mathbf{M}\\
\mathbf{0} & \mathbf{M} - \mathbf{G}
\end{array}
\right) \text{ ,} \quad \text{where} \quad
\mathbf{G} = 
\left(
\begin{array}{ c c c c }
\gamma & 0 & \cdots & 0 \\
0 & \ddots & \ddots & \vdots \\
\vdots & \ddots & \ddots & 0 \\
0 & \cdots & 0 & \gamma
\end{array}
\right) \text{ .}
\end{equation}
Then the early behaviour of the epidemic is given by setting $\mathbf{I}(0) =
\mathbf{0}+\varepsilon \mathbf{y}$, leading to early dynamics of the form
\begin{equation} 
\frac{{\rm d}}{{\rm d}t} \mathbf{I}(t) =
\left(\mathbf{M} - \mathbf{G} \right) \mathbf{I}(t)
+ O(\varepsilon^2) \text{ .} \label{earlydyn}
\end{equation}

\subsection{Perron-Frobenius analysis}

In general, we expect the matrix $\mathbf{M} - \mathbf{G}$ to
be quasi-positive, since we expect every age to interact with every other at
some level, leading to positive off-diagonal elements, but recovery can lead to
elements $\leq 0$ on the diagonal. We consider the special case where each age
class is capable of supporting the disease independently of the others (i.e.\
$\beta_{a,a} > \gamma$) which is more likely to apply if the disease is more
transmissible and there are fewer age classes in the model. In this special
case, the matrix $\mathbf{M} - \mathbf{G}$ has positive real
entries, allowing us to use the following important theorem.  \\
\begin{thm} 
\emph{(Perron-Frobenius)}
Let $\mathbf{A}=(a_{i,j})$ be an n$\times$n positive matrix: $a_{i,j} > 0$ for
$1 \leq i,j \leq n$. Then the following statements hold:
\begin{enumerate}
\item There is a positive real number $\lambda_1$, called the {\bf Perron root}
or the {\bf Perron-Frobenius eigenvalue}, such that $\lambda_1$ is an eigenvalue of
$\mathbf{A}$ and any other eigenvalue $\lambda_{i\neq 1}$ has
$| \lambda_i | < \lambda_1$.
\item The Perron-Frobenius eigenvalue is simple.
\item There exists an eigenvector $\mathbf{v}_1$ of $\mathbf{A}$ with eigenvalue
$\lambda_1$ such that all components of $\mathbf{v}_1$ are
positive.
\item There are no other positive eigenvectors. I.e.\ all other eigenvectors
must have at least one negative or non-real component.
\end{enumerate}
\end{thm}
\begin{proof}
See e.g.\ \citet[Chapter 8]{Meyer00}.
\end{proof}

We will also assume that the initial vector direction for the dynamical system,
$\mathbf{y}$, can be written as a linear combination of the eigenvectors
$\{\mathbf{v}_i\}$ of $\mathbf{M} - \mathbf{G}$
\begin{equation}
\mathbf{y} = 
\sum_{i=1}^{n} C_i \mathbf{v}_{i} \text{ ,}
\qquad (C_i \in \mathbb{C}) \text{ .}
\end{equation}
This, together with~\eqref{earlydyn}, means that early in the epidemic,
\begin{equation}
\begin{aligned}
\mathbf{I}(t) & = \sum_{i=1}^{n} c_i \mathbf{v}_{i} {\rm e}^{\lambda_i t}
\text{ ,} \qquad (c_i \in \mathbb{C})\\
& = {\rm e}^{\lambda_1 t}\left(c_1 \mathbf{v}_1  +
\sum_{i=2}^{n} c_i \mathbf{v}_{i} {\rm e}^{(\lambda_i - \lambda_1) t} \right)
\text{ .} \label{ivexpansion}
\end{aligned}
\end{equation}
From the Perron-Frobenius theorem, we can then argue that for sufficiently
small $|c_{i>1}|$, and sufficiently large $t$, the contribution to $\mathbf{I}$
from eigenvectors other than $\mathbf{v}_1$ can be made arbitrarily small since
the absolute value of such contributions decays exponentially over time.

\section{Analytical approach}

\subsection{Models of assortative mixing}

We now consider mixing matrices that have the assortative property observed
in real networks, but remain amenable to analytic methods. These are in fact
special cases of the matrices considered using a time-independent framework to
determine invasion thresholds in \citet[\S{}5.3.2]{Diekmann00}. Our aim,
however, is to consider time-dependent transient behaviour for diseases that
are successfully invading.\\

\begin{definition} 
A \emph{Basic Mixing Matrix} is an $n\times n$ matrix $\mathbf{B}$ such that
\begin{equation*}
\mathbf{B}(\alpha,n) \coloneqq \left(
\begin{array}{ c c c c c }
1 + \alpha & 1 & 1 & \cdots & 1 \\
1 & 1 + \alpha & 1 & \cdots & 1 \\
1 & 1 & 1 + \alpha & \ddots & \vdots \\
\vdots & \vdots & \ddots & \ddots & 1 \\
1 & 1 & \cdots & 1 & 1+ \alpha 
\end{array}
\right)
\end{equation*}
Where $ \alpha \in \mathbb{R} $, $ \alpha > -1 $.
\end{definition}

\begin{thm} 
Let $ \mathbf{B}(\alpha,n)$ be a Basic Mixing Matrix, then $ \lambda = n +
\alpha $ is an eigenvalue with eigenvector $ \mathbf{v}_1=(1,1,\ldots,1)^{\top}
$ and $ \tilde{\lambda} = \alpha $ is an eigenvalue with algebraic multiplicity
$ n-1 $ corresponding to $ n-1 $ eigenvectors of the form; \newline $$
\mathbf{v}_{i\neq 1} = (1,1,\ldots,1-n,\ldots,1)^{\top}, $$ with $ 1-n $ in the
$ i^{th} $ place. \label{thm:bmm}
\end{thm}
\begin{proof}
Eigenvalues and eigenvectors satisfy the equation $ \mathbf{B}\mathbf{v} =
\lambda \mathbf{v}$, hence letting $ \mathbf{v}_1=(1,1,\ldots,1)^{\top}$ gives
$ \lambda = \lambda_1 = n+ \alpha $, and by Perron-Frobenius, this eigenvalue is
simple and all other eigenvalues are less than $ \lambda $ with eigenvectors
that have at least one negative component.  If we now take $ \mathbf{v}_{i} =
(1,1,\ldots,1-n,\ldots,1)^{\top}, \ i = 1,\ldots, n $ for $ i \neq 1 $, where $
1-n $ is in the $ i^{th} $ place, this gives $ \lambda_{i} = \alpha, \ \forall i
\neq 1 $. Then since $ \mathbf{v}_{i} $ are linearly independent $ \forall i $
we have found the entire system of $ n $ eigenvectors and $ \tilde{\lambda} =
\lambda_{i\neq 1} = \alpha $ has algebraic multiplicity $ n-1 $.
\end{proof}

\begin{definition} 
A \emph{Simple Mixing Matrix} is an $n\times n$ matrix $\mathbf{S}$ such that
\begin{equation*}
\mathbf{S}(\boldsymbol{\alpha}, n) \coloneqq \left(
\begin{array}{ c c c c c }
1 + \alpha_{1} & 1 & 1 & \cdots & 1 \\
1 & 1 + \alpha_{2} & 1 & \cdots & 1 \\
1 & 1 & 1 + \alpha_{3} & \ddots & \vdots \\
\vdots & \vdots & \ddots & \ddots & 1 \\
1 & 1 & \cdots & 1 & 1+ \alpha_{n}
\end{array}
\right)
\end{equation*}
Where $ \alpha_i \in \mathbb{R} $, $ \alpha_i > -1 $.
\end{definition}

\begin{thm} 
The Simple Mixing Matrix $\mathbf{S}(\boldsymbol{\alpha}, 2)$ has eigensystem
\begin{equation}
\lambda_{\pm} = \frac{2+ \alpha_{1} + \alpha_{2} \pm \sqrt{(\alpha_{1} - \alpha_{2})^2 +4}}{2} \text{ ,}\qquad \mathbf{v}_{\pm} = \left( 
	\begin{array}{ c }
	1 \\
	 \frac{ \alpha_{2} - \alpha_{1} \pm \sqrt{(\alpha_{1} - \alpha_{2})^2 +4}}{2}
	\end{array}
\right) \text{ .} \label{s2esyst}
\end{equation}
\end{thm}
\begin{proof}
Substitute~\eqref{s2esyst} into $\mathbf{S}(\boldsymbol{\alpha}, 2)
\mathbf{v}_{\pm} = \lambda_{\pm} \mathbf{v}_{\pm}$.
\end{proof}

Note that finding the eigensystem of $\mathbf{S}(\boldsymbol{\alpha}, n)$ in
general involves solving an order-$n$ polynomial, but no simple general formula
exists as for a Basic Mixing Matrix, limiting the usefulness of analytic
methods, and we introduce this concept mainly to explain why a more general
analytical approach is not possible.

\subsection{Early time estimate}

We now wish to define a timescale over which the epidemic approaches its dominant
eigenvalue.\\

\begin{definition} \label{taudef} Consider an epidemic model where the early
vector of infectious individuals is given by~\eqref{ivexpansion}. The
\emph{Early Time} is defined as
\begin{equation}
\tau \coloneqq \frac{1}{\lambda - \tilde{\lambda}} \text{ ,}
%\end{equation}
\qquad \text{where} \quad
%\begin{equation}
\tilde{\lambda} \coloneqq
\mathrm{max}\left\{\mathrm{Re}(\lambda_{i>1})\right\}
\text{ .}
\end{equation}
\end{definition}

This definition has the property that over a period of, say, time $2\tau$, the
relative importance of each non-dominant eigenvector reduces by a factor of at
least $\mathrm{e}^2 \approx 7.389$.

\subsection{Late time estimate}

An estimation of the late time is also necessary to assess if convergence to
the dominant eigenvalue happens before susceptibles are depleted. Making use of
the standard definition of $R_0$ \citep{Diekmann:1990} as the dominant
eigenvalue of the next-generation matrix, in our case
$\mathbf{M}/\gamma$, we expect intuitively that the epidemic is near to
its maximum when the proportion of the population susceptible is $ S(t_*) =
1/R_0 $, since for this susceptible population (if allocated appropriately to
different age classes) the disease can no longer invade.\\

\begin{definition} \label{Tdef} Consider an epidemic model with basic
reproductive ratio $R_0$ given by the dominant eigenvalue of
$\mathbf{M}/\gamma$, starting with a proportion $\varepsilon$ of the
population infectious. The \emph{Late Time} is defined as
\begin{equation}
T \coloneqq  \frac{1}{\lambda} \ln \left( \frac{R_0 - 1}{R_0 \varepsilon} \right)
\text{ .}
\end{equation}
\end{definition}
This definition makes sense since if the population is initially
entirely susceptible we may approximate $ S(t) \approx 1 -
\varepsilon {\rm e}^{\lambda t} $. The time at which this quantity first exceeds
$1/R_0$ should therefore correspond to a sensible definition of late time.
Figure~\ref{fig:latetimeapprox} visualises the concepts involved in this
definition.

\subsection{Separation of early and late time}

We are now in a position to compare early and late time. Clearly, it is
always possible to make $\varepsilon$ sufficiently small to keep these
timescales separated; explicitly, this requires
\begin{equation}
\tau \ll T \quad \Rightarrow \quad \varepsilon \ll 
 \frac{R_0 - 1}{R_0}\; \mathrm{exp}\left( \frac{\lambda}%
{\tilde{\lambda} - \lambda} \right)
\label{lleps} \text{ .}
\end{equation}
Now suppose we consider transmission based on a Basic Mixing Matrix
multiplied by an overall transmission rate $\nu$ so that 
$\mathbf{M} = \nu \mathbf{B}(\alpha,n)$, meaning that the eigensystem
of $\mathbf{M} - \mathbf{G}$ can be obtained by application of
Theorem~\ref{thm:bmm}.  The basic reproductive ratio is $R_0 =
\nu(n+\alpha)/\gamma$, and for an initial infection of $ I_{0} $ in the $
i^{th} $ age group, 
\begin{equation}
\mathbf{I}(0) = ( 0, \ldots, 0, \underbracket{I_0}_{%
\mathclap{i\text{-th position}}}, 0, \ldots, 0 )^{\top}
= \frac{I_{0}}{n}\ (\mathbf{v}_1 - \mathbf{v}_{i}) \text{ .}
\end{equation} 
Note that the symmetry of the system means that we can let $i\neq 1$ without
loss of generality. This initial condition leads to early behaviour
\begin{equation} \label{eq:eigensystem}
\mathbf{I}(t) = \frac{ I_0 }{n} {\rm e}^{ \lambda t} \left( \mathbf{v}_1 -
{\rm e}^{-(\lambda - \tilde{\lambda})t} \mathbf{v}_i \right) \text{ ,}
\quad \text{where} \quad
\lambda = \nu (n + \alpha) - \gamma \text{ ,}\quad
\tilde{\lambda} = \nu \alpha - \gamma \text{ ,} \quad
\end{equation} 
and so the early time $\tau = 1/(\nu n)$ .  It is most instructive to consider
what happens at constant $R_0$, so substituting into~\eqref{lleps} in this
regime gives separation of timescales when 
\begin{equation}
I_0 \ll 
 \frac{R_0-1}{R_0}
\mathrm{exp}\left(\frac{R_0 -1}{(\nu\alpha/\gamma) - R_0} \right)
\text{ .} \label{i0ll}
\end{equation}
Taking appropriate partial derivatives of the right-hand side  of this
expression, we see that for a given $R_0$, a lower $I_0$ is required to secure
convergence if we increase $\nu$ or $\alpha$, or if we decrease $\gamma$. This
is as would be expected from thinking qualitatively about the problem.

\section{Numerical approach}

\subsection{Quantification of convergence}

It will be convenient to define a numerical measure $\phi(t)$ of convergence to the
dominant eigenvalue. 
\begin{equation}
\begin{aligned}
\label{psi}
\phi(t) & \coloneqq \mathrm{tan}^{-1} ( \psi(t) ) \text{ ,}\\ \text{where}\qquad
\psi(t) & \coloneqq \frac{{\rm d}}{{\rm d}t} \displaystyle \left( \mathrm{ln} \left( \frac{
I(t) }{ J(t) } \right)\right) \text{ ,}\\
I(t) & \coloneqq \sum_a I_a(t) \text{ ,}\\
J(t) & \coloneqq {\rm e}^{\lambda t} \text{ .}
\end{aligned}
\end{equation}
This definition gives $\phi(t) = 0$ at times when the epidemic moves along the
dominant eigenvector, and non-zero values in the interval $(-1,1)$ otherwise.

\subsection{Contact survey data}

We take the form for $\mathbf{M}$ in our numerical work from the
results of the POLYMOD study \citep{Mossong08} for all reported
contacts (physical and conversational) across all countries.  We use
age classes of width 5 years and a final class for ages 70+. Consideration of
population pyramids \citep{Stats07,Stillwell11} shows that each age group is of
almost equal weighting up to 70, after which the rest of the population is in a
rapidly declining tail which in total is roughly equal to the weighting of one
of the 5 year classes, and so for present purposes where we are looking for
broad insights rather than precise predictions the normalisation~\eqref{norm}
is appropriate.

\subsection{Results}

When we integrate~\eqref{SIR} directly using Runge-Kutta, for
infection starting in the 15--19 age group, Figure \ref{fig:phipoly} shows the
results obtained for $\phi(t)$ and $I(t)$. Strikingly, this plot shows that
even during periods of the epidemic where the prevalence timeseries looks
linear on a logarithmic y-axis, $\phi(t)$ can fail to equal zero for an
appreciable period of time (meaning that the dynamics~\eqref{SIR} are not
governed by the dominant eigenvector) for modest initial infectious
populations.  Whether such initial concentrations of infection can be seen,
perhaps as a result of an unpredictable event early in the epidemic when
stochastic effects dominate, depends on the population size that can be seen as
mixing according to~\eqref{SIR}. At country size -- i.e.\ tens or hundreds of
millions of individuals -- small $I_0$ is most likely, but at city sizes of
around $10^5$, it is not inconceivable that an epidemic could start with
several hundred students in the same school leading to $I_0 \sim 10^{-3}$.

\section{Discussion}

In this paper, we have considered the real-time rate of convergence of
age-structured SIR epidemic models to their dominant eigenvector, both
analytically for a basic model of assortative mixing, and numerically for
realistic mixing.  In each case, sufficiently small initially infectious
populations are needed to ensure that the dynamical system does indeed converge
before the depletion of susceptibles becomes important. In many cases this
effect will not matter; but it is plausible that for some supercritical
epidemics in moderately-sized, closed populations with strongly age-structured
mixing, this failure of convergence to the dominant eigenvector will be a
source of bias. In particular, attempts to infer the relative susceptibility of
different age classes may become biased. If a particular age group is
over-represented in the initially infected group as compared to the dominant
eigenvector, then its relative susceptibility may be overestimated.  Correction
for non-convergence to the dominant eigenvector is therefore potentially
important in epidemiological analysis.

\section*{Acknowledgements}

TH is supported by the Engineering and Physical Sciences Research Council.

\begin{figure}[H]
\centering
\includegraphics[width=0.9\textwidth]{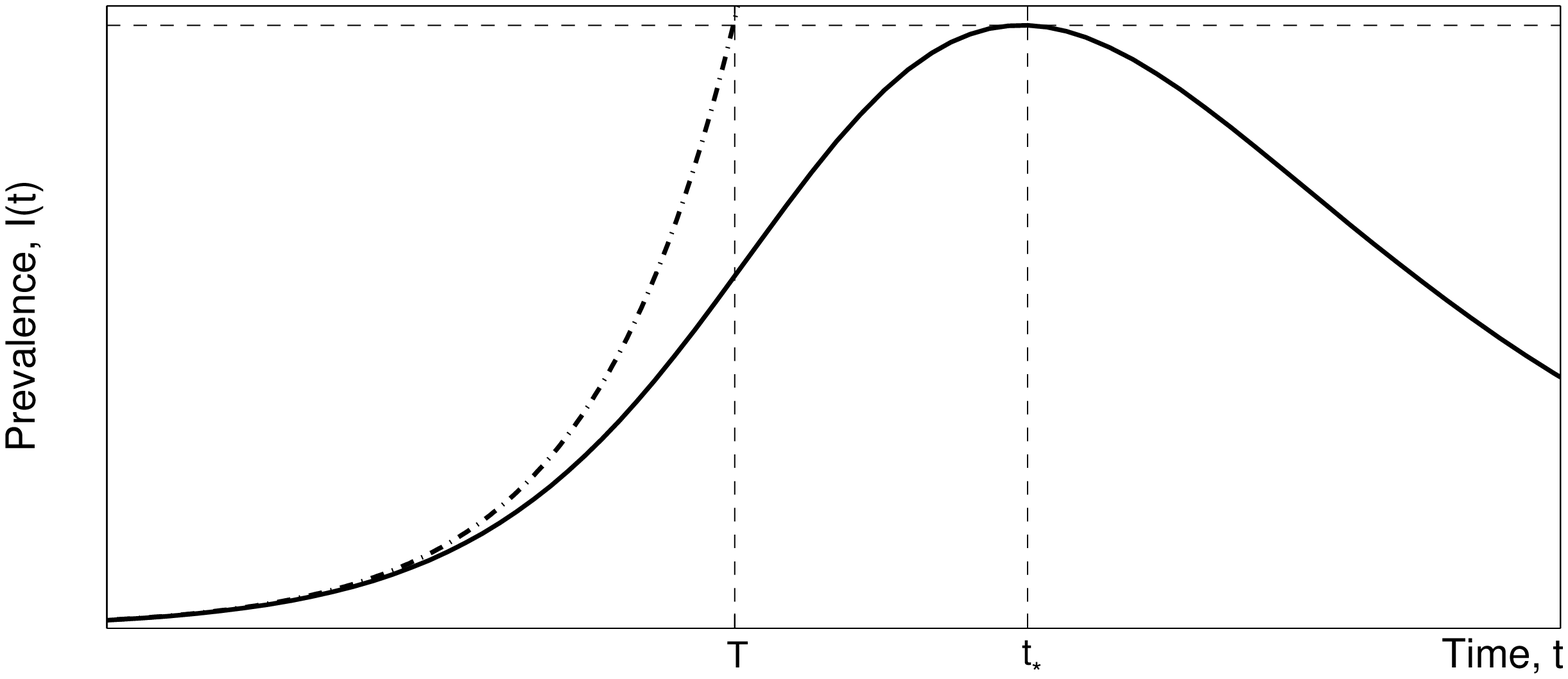}
\caption{Visualisation of the late time approximation. Peak time $t_*$
corresponds to the peak of the full epidemic (solid black line) and late time
$T$ is the equivalent timescale for the exponential growth approximation
(dashed line).}
\label{fig:latetimeapprox}
\end{figure}

\begin{figure}[H]
\centering
\includegraphics[width=0.9\textwidth]{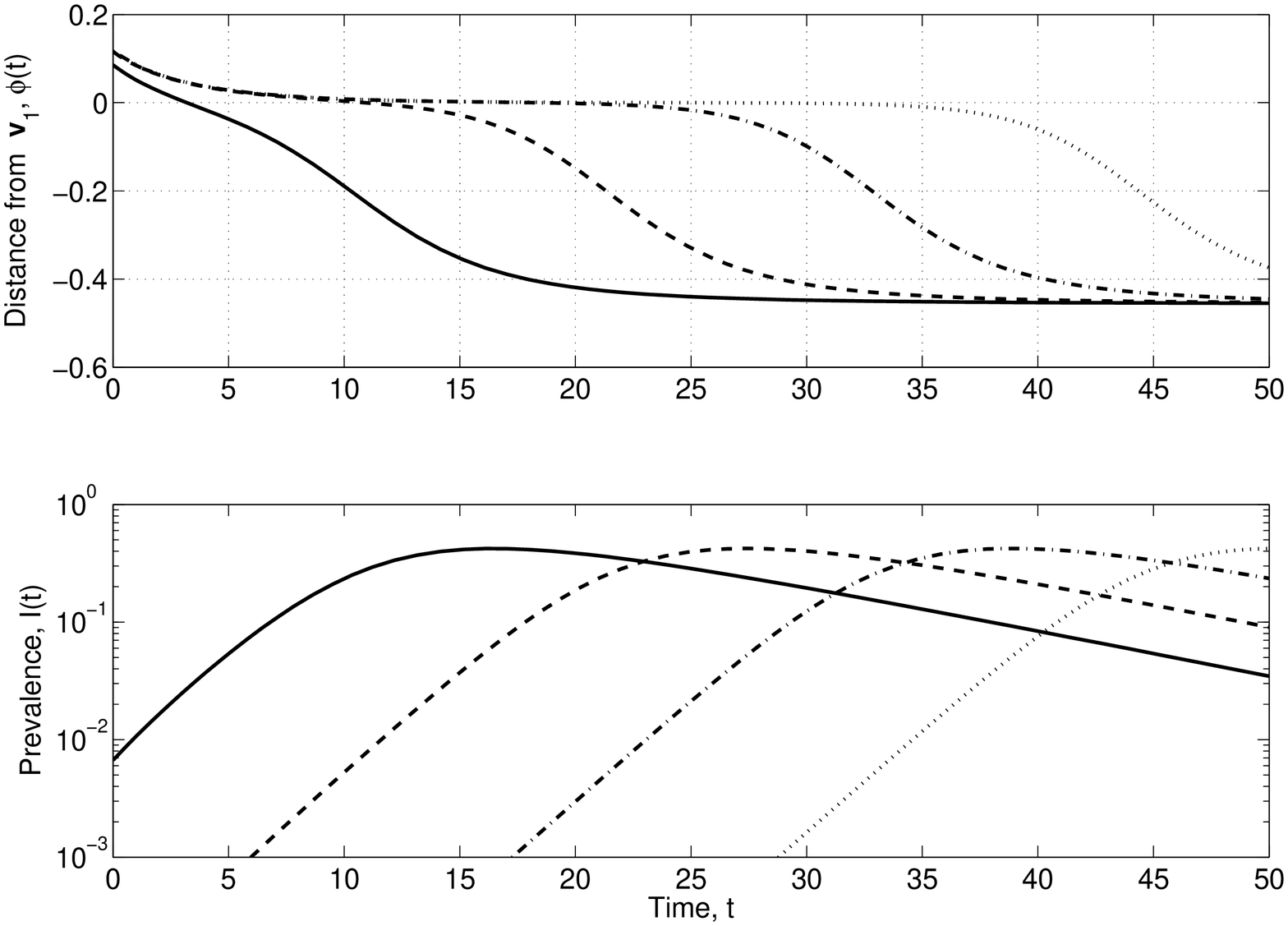}
\caption{Top: Distance from the dominant eigenvalue, $\phi$, for the POLYMOD
mixing matrix. Bottom: prevalence against time on a logarithmic scale.
Parameters are: $R_0=5$, $\gamma=0.1$, $I_0=10^{-1}$ (solid lines),
$I_0=10^{-3}$ (dashed lines), $I_0=10^{-5}$ (dash-dot lines), $I_0=10^{-7}$
(dotted lines).  Infection is started in the 15--19 year-old age
group. }
\label{fig:phipoly}
\end{figure}

\end{document}